\documentclass[12pt]{article}
\usepackage[margin=1in]{geometry}

\usepackage{graphicx}
\usepackage{hyperref}
\usepackage{url}
\usepackage{amsthm}
\usepackage{amssymb}
\usepackage{amsmath}
\usepackage{mathtools}
\usepackage{xspace}
\usepackage{enumitem}
\usepackage{color}
\usepackage[all]{xy}
\usepackage{lineno}
\usepackage{anyfontsize}

\newcommand{\diam}{\mathrm{diam}}

\newcommand{\ri}{\mathcal{R}}
\newcommand{\rin}{\mathcal{R}^{\infty}}
\newcommand{\eps}{\varepsilon}
\newcommand{\Cech}{\v{C}ech\xspace}

\newcommand{\R}{\mathbb{R}}
\newcommand{\Z}{\mathbb{Z}}
\newcommand{\N}{\mathbb{N}}

\newcommand{\us}{\mathcal{ULS}}

\newcommand{\ux}{\mathcal{X}}

\newcommand{\tg}{\tilde{g}}
\newcommand{\ch}{\mathcal{C}}
\newcommand{\lin}{L_\infty}
\newcommand{\dmn}{\diam_\infty}

\newcommand{\field}{\mathcal{F}}
\newcommand{\kernel}{\mathrm{ker}}
\newcommand{\image}{\mathrm{im}}

\newcommand{\simplicialmap}{\varphi}
\newcommand{\chainmap}{\phi}
\newcommand{\linearmap}{\lambda}

\DeclarePairedDelimiter{\ceil}{\lceil}{\rceil}

\newtheorem{lemma}{Lemma}

\newtheorem{theorem}[lemma]{Theorem}

\begin{document}

\title{Improved Approximate Rips Filtrations with \\Shifted Integer Lattices}
\author{
Aruni Choudhary\footnote{Max Planck Institute for Informatics, 
Saarbr\"ucken, Germany \texttt{(aruni.choudhary@mpi-inf.mpg.de)}}
\and 
Michael Kerber\footnote{Graz University of Technology, 
Graz, Austria\texttt{(kerber@tugraz.at)}} 
\and
Sharath Raghvendra\footnote{Virginia Tech, 
Blacksburg, USA \texttt{(sharathr@vt.edu)}}}

\maketitle 

\begin{abstract}
Rips complexes are important structures for analyzing topological features of 
metric spaces. 
Unfortunately, generating these complexes constitutes an expensive task
because of a combinatorial explosion in the complex size.
For $n$ points in $\mathbb{R}^d$,
we present a scheme to construct a $3\sqrt{2}$-approximation of the
multi-scale filtration of the $L_\infty$-Rips complex, which extends to
a $O(d^{0.25})$-approximation of the Rips filtration for the Euclidean case.
The $k$-skeleton of the resulting approximation 
has a total size of $n2^{O(d\log k)}$. 
The scheme is based on the integer lattice 
and on the barycentric subdivision of the $d$-cube.
\end{abstract}

\section{Introduction}
\label{section:intro}

\emph{Persistent homology}~\cite{carlsson-survey,brunrer-book,elz-topological} is a 
technique to analyze of data sets using topological invariants. 
The idea is to build a multi-scale representation of the data set
and to track its homological changes across the scales.

A standard construction for the important case of point clouds in
Euclidean space is the \emph{Vietoris-Rips complex} 
(or just \emph{Rips complex}): for a scale parameter $\alpha\ge 0$, it is 
the collection of all subsets of points with diameter at most $\alpha$.
When $\alpha$ increases from $0$ to $\infty$, the Rips complexes form
a \emph{filtration}, an increasing sequence of nested simplicial complexes
whose homological changes can be computed and represented in terms of a 
\emph{barcode}.

The computational drawback of Rips complexes is their sheer size:
the $k$-skeleton of a Rips complex (that is, only subsets of size $\leq k+1$
are considered) for $n$ points consists of $\Theta(n^{k+1})$ simplices
because every $(k+1)$-subset joins the complex for a sufficiently large
scale parameter. 
This size bound turns barcode computations for large point clouds
infeasible even for low-dimensional homological features\footnote{An exception
are point clouds in $\R^2$ and $\R^3$, 
for which \emph{alpha complexes}~\cite{brunrer-book}
are an efficient alternative.}.
This poses the question of what we can say
about the barcode of the Rips filtration 
without explicitly constructing all of its simplices.

We address this question using approximation techniques. 
Barcodes form a metric space: two barcodes are close if the same 
homological features occur on roughly the same range of scales 
(see Section~\ref{section:backgnd} for the precise definition).
The first approximation scheme by Sheehy~\cite{sheehy-rips}
constructs a $(1+\eps)$-approximation of the $k$-skeleton of the 
Rips filtration
using only $n(\frac{1}{\eps})^{O(\lambda k)}$ simplices for arbitrary finite
metric spaces, 
where $\lambda$ is the doubling dimension of the metric.
Further approximation techniques for Rips complexes~\cite{desu-gic} 
and the closely related \emph{\Cech 
complexes}~\cite{bs-approximating,sheehy-cech,kbs-cech} have been derived 
subsequently, all with 
comparable size bounds.
More recently, we constructed an approximation scheme for Rips complexes
in Euclidean space that yields a worse approximation factor of $O(d)$, but uses
only $n2^{O(d\log k)}$ simplices~\cite{ckr-polynomial}, where $d$ is the ambient
dimension of the point set.

\subparagraph*{Contributions}
We present a $3\sqrt{2}$-approximation for the Rips filtration of 
$n$ points in $\R^d$ in the $L_\infty$-norm ,
whose $k$-skeleton has size $n2^{O(d\log k)}$. 
This translates to a $O(d^{0.25})$-approximation of the 
Rips filtration in the Euclidean metric and hence improves the asymptotic
approximation quality of our previous approach~\cite{ckr-polynomial} 
with the same size bound.

On a high level, our approach follows a straightforward approximation scheme:
given a scaled and appropriately shifted integer grid on $\R^d$, 
we identify those grid points that are close to the input points and 
build an approximation complex using these grid points.
The challenge lies in how to connect these grid points to a simplicial complex
such that close-by grid points are connected,
while avoiding too many connections to keep the size small. 
Our approach first selects a set of \emph{active faces}
in the cubical complex defined over the grid, and defines the approximation
complex using the barycentric subdivision of this cubical complex.

We also describe an output-sensitive algorithm to compute our approximation.
By randomizing the aforementioned shifts of the grids, we obtain a
worst-case running time of $n2^{O(d)}\log\Delta+2^{O(d)}M$,
where $\Delta$ is \emph{spread} of the point set (that is, the ratio
of the diameter to the closest distance of two points)
and $M$ is the size of the approximation.

Additionally, this paper makes the following technical contributions:

\begin{itemize}
\item 

We follow the standard approach of defining a sequence
of approximation complexes
and establishing an \emph{interleaving} between the Rips filtration 
and the approximation.
We realize our interleaving using \emph{chain maps} connecting 
a Rips complex at scale $\alpha$ to an approximation complex at scale $c\alpha$, 
and vice versa, with $c\geq 1$ being the approximation factor. 
Previous approaches~\cite{ckr-polynomial,desu-gic,sheehy-rips} used
\emph{simplicial maps} for the interleaving, which induce an elementary
form of chain maps and are therefore more restrictive.

The explicit construction of such maps can be a non-trivial task. 
The novelty of our approach is that we avoid this construction 
by the usage of \emph{acyclic carriers}~\cite{munkres}.
In short, carriers are maps that assign
subcomplexes to subcomplexes under some mild extra conditions.
While they are more flexible, they still certify the existence of 
suitable chain maps, as we exemplify in Section~\ref{section:scheme}.
We believe that this technique is of general interest for the construction
of approximations of cell complexes.

\item We exploit a simple trick that we call \emph{scale balancing} to improve
the quality of approximation schemes. 
In short, if the aforementioned interleaving maps from and to the Rips filtration 
do not increase the scale parameter by the same amount, one can simply multiply 
the scale parameter of the approximation by a constant. 
Concretely, given maps
\[
\phi_\alpha:\ri_\alpha\rightarrow \mathcal{X}_\alpha\qquad 
\psi_\alpha:\mathcal{X}_\alpha\rightarrow \ri_{c\alpha}
\]
interleaving the Rips complex $\ri_\alpha$ and the approximation 
complex $\mathcal{X}_\alpha$,
we can define $\mathcal{X}'_\alpha:=\mathcal{X}_{\alpha/\sqrt{c}}$ and obtain maps
\[\phi'_\alpha:\ri_\alpha\rightarrow \mathcal{X}'_{\sqrt{c}\alpha}\qquad 
\psi_\alpha:\mathcal{X}'_\alpha\rightarrow \ri_{\sqrt{c}\alpha}\]
which improves the interleaving from $c$ to $\sqrt{c}$.
While it has been observed that the same trick can be used for improving
the worst-case distance between Rips and \Cech filtrations\footnote{Ulrich Bauer, private communication}, our work seems to be the first to make use of it in
the context of approximations of filtrations.
\end{itemize}

Our technique can be combined with
dimension reduction techniques in the same way as in~\cite{ckr-polynomial}
(see Theorems 19, 21, and 22 therein), with improved logarithmic factors.
We omit the technical details in this paper. 
Also, we point out that 
the complexity bounds for size and computation time are for the entire
approximation scheme and not for a single scale as in~\cite{ckr-polynomial}.
However, similar techniques as the ones exposed 
in Section~\ref{section:computational} can be used to improve the results
of~\cite{ckr-polynomial} to hold for the entire approximation as 
well\footnote{An extended version of~\cite{ckr-polynomial}
containing these improvements is currently under submission.}.

\subparagraph*{Outline} We start the presentation by discussing the relevant 
topological concepts in Section~\ref{section:backgnd}. 
Then, we present few results about grid lattices in Section~\ref{section:grids}. 
Building on these ideas, the approximation scheme is presented in 
Section~\ref{section:scheme}. 
Computational aspects of 
the approximation scheme are discussed in Section~\ref{section:computational}. 
We conclude in Section~\ref{section:conclusion}.

\section{Background}
\label{section:backgnd}

We review the topological concepts needed in our argument;
see~\cite{bss-metrics,ch-proximity,brunrer-book,munkres} for more details.

\subparagraph*{Simplicial complexes}

A \emph{simplicial complex} $K$ on a finite set of elements $S$ 
is a collection of subsets $\{\sigma\subseteq S\}$ called \emph{simplices} 
such that each subset $\tau\subset\sigma$ is also in $K$.
The dimension of a simplex $\sigma\in K$ is $k:=|\sigma|-1$, 
in which case $\sigma$ is called a \emph{$k$-simplex}.
A simplex $\tau$ is a \emph{subsimplex} of $\sigma$ if $\tau\subseteq\sigma$. 
We remark that, commonly a subsimplex is called a 'face' of a simplex,
but we reserve the word 'face' for a different structure.
For the same reason, we do not introduce the common notation of
of 'vertices' and 'edges' of simplicial complexes, but rather refer
to $0$- and $1$-simplices throughout.
The \emph{$k$-skeleton} of $K$ consists of
all simplices of $K$ whose dimension is at most $k$.
For instance, the $1$-skeleton of $K$ is a graph 
defined by its $0$-simplices and $1$-simplices.

Given a point set $P\subset\R^d$ and a real number $\alpha\ge 0$,
the \emph{(Vietoris-)Rips} complex on $P$ at scale $\alpha$ consists of all
simplices $\sigma=(p_0,\ldots,p_k)\subseteq P$ such that $diam(\sigma)\le \alpha$,
where $diam$ denotes the diameter.
In this work, we write $\ri_\alpha$ for the Rips complex at scale $\alpha$ 
with the Euclidean metric,
and $\rin_\alpha$ when using the metric of the $L_\infty$-norm.
In either way, a Rips complex is an example of a \emph{flag complex},
which means that whenever a set $\{p_0,\ldots,p_k\}\subseteq P$ has the property
that every $1$-simplex $\{p_i,p_j\}$ is in the complex, then the
$k$-simplex $\{p_0,\ldots,p_k\}$ is also in the complex.

A simplicial complex $K'$ is a \emph{subcomplex} of $K$ if $K'\subseteq K$.
For instance, $\ri_{\alpha}$ is a subcomplex of $\ri_{\alpha'}$ for 
$0\le\alpha\leq\alpha'$. 
Let $L$ be a simplicial complex.
Let $\hat{\simplicialmap}$ be a map which assigns to each vertex of $K$, 
a vertex of $L$.
A map $\simplicialmap:K\rightarrow L$ is called a \emph{simplicial map}
induced by $\hat{\simplicialmap}$,
if for every simplex $\{p_0,\ldots,p_k\}$ in $K$, the set 
$\{\hat{\simplicialmap}(p_0),\ldots,\hat{\simplicialmap}(p_k)\}$
is a simplex of $L$. 
For $K'$ a subcomplex of $K$, the inclusion map $inc:K'\rightarrow K$
is an example of a simplicial map. 
A simplicial map $K\rightarrow L$ is completely determined
by its action on the $0$-simplices of $K$.

\subparagraph*{Chain complexes}

A \emph{chain complex} $\ch_\ast=(\ch_p,\partial_p)$ with $p\in\N$ is a collection 
of abelian groups $\ch_p$ and homomorphisms $\partial_p:\ch_p\rightarrow \ch_{p-1}$ 
such that $\partial_{p-1}\circ\partial_{p}=0$.
A simplicial complex $K$ gives rise to a chain complex $\ch_\ast(K)$ by fixing a 
base field $\field$, defining $\ch_p$ as the set of formal linear combinations of 
$p$-simplices in $K$ over $\field$, and $\partial_p$ as the linear operator that 
assigns to each simplex the (oriented) sum of its sub-simplices of codimension 
one\footnote{To avoid thinking about orientations, it is often assumed
that $\field=\Z_2$ is the field with two elements.}.

A \emph{chain map} $\chainmap:\ch_\ast\rightarrow D_\ast$ between
chain complexes $\ch_\ast=(\ch_p,\partial_p)$ and $D_\ast=(D_p,\partial'_p)$
is a collection of group homomorphisms $\chainmap_p:\ch_p\rightarrow D_p$
such that $\chainmap_{p-1}\circ\partial_{p}=\partial'_{p}\circ\chainmap_{p}$.
For example, a simplicial map $\simplicialmap$ between simplicial complexes induces
a chain map $\bar{\simplicialmap}$ between the corresponding chain complexes.
This construction is \emph{functorial}, meaning that for 
$\simplicialmap$ the identity function on a simplicial complex $K$, 
$\bar{\simplicialmap}$ is the identity function on $\ch_\ast(K)$,
and for composable simplicial maps $\simplicialmap,\simplicialmap'$,
we have that 
$\overline{\simplicialmap\circ\simplicialmap'}=\bar{\simplicialmap}\circ\bar{\simplicialmap'}$.

\subparagraph*{Homology and carriers}

The \emph{$p$-th homology group} $H_p(\ch_\ast)$ of a chain complex is defined 
as $\kernel\,\partial_p/\image\,\partial_{p+1}$.
The $p$-th homology group of a simplicial complex $K$, $H_p(K)$, is the 
$p$-th homology group of its induced chain complex.
In either case $H_p(\ch_\ast)$ is a $\field$-vector space because we have 
chosen our base ring $\field$ as a field.
Intuitively, when the chain complex is generated from a simplicial complex, the 
dimension of the $p$-th homology group counts the number of $p$-dimensional holes 
in the complex 
(except for $p=0$, where it counts the number of connected components).
We write $H(\ch_\ast)$ for the direct sum of all $H_p(\ch_\ast)$ for $p\geq 0$.

A chain map $\chainmap:\ch_\ast\rightarrow D_\ast$ 
induces a linear map $\chainmap^\ast: H(\ch_\ast)\rightarrow H(D_\ast)$
between the homology groups. 
Again, this construction is functorial, meaning that it maps identity maps to 
identity maps, and it is compatible with compositions.

We call a simplicial complex $K$ \emph{acyclic}, if $K$ is connected and all 
homology groups $H_p(K)$ with $p\geq 1$ are trivial.
For simplicial complexes $K$ and $L$, 
an \emph{acyclic carrier} $\Phi$ is a map that assigns to 
each simplex $\sigma$ in $K$, a non-empty subcomplex $\Phi(\sigma)\subseteq L$ 
such that $\Phi(\sigma)$ is acyclic, and whenever $\tau$ is a subsimplex 
of $\sigma$, then $\Phi(\tau)\subseteq\Phi(\sigma)$.
We say that a chain $c\in\ch_p(K)$ is \emph{carried} by a subcomplex $K'$, 
if $c$ takes value $0$ except for $p$-simplices in $K'$.
A chain map $\chainmap:\ch_\ast(K)\rightarrow \ch_\ast(L)$ is 
\emph{carried by $\Phi$}, if for each simplex $\sigma\in K$,
$\chainmap(\sigma)$ is carried by $\Phi(\sigma)$.
We state the \emph{acyclic carrier theorem}~\cite{munkres}:
\begin{theorem}
\label{theorem:acyclic_carrier}
Let $\Phi:K\rightarrow L$ be an acyclic carrier.
\begin{itemize}
	\item There exists a chain map $\chainmap:\ch_\ast(K)\rightarrow \ch_\ast(L)$ 
	such that $\chainmap$ is carried by $\Phi$.
	\item If two chain maps $\chainmap_1,\chainmap_2:\ch_\ast(K)\rightarrow \ch_\ast(L)$
	are both carried by $\Phi$, then $\chainmap_1^\ast=\chainmap_2^\ast$.
\end{itemize}
\end{theorem}

\subparagraph*{Filtrations and towers}
Let $I\subseteq\R$ be a set of real values which we refer to as \emph{scales}.
A \emph{filtration} is a collection of simplicial complexes
$(K_\alpha)_{\alpha\in I}$ such that $K_\alpha\subseteq K_\alpha'$ 
for all $\alpha\leq\alpha'\in I$. 
For instance, $(\ri_\alpha)_{\alpha\geq 0}$ is a filtration 
which we call the \emph{Rips filtration}.
A \emph{(simplicial) tower} is a sequence $(K_\alpha)_{\alpha\in J}$ of simplicial 
complexes with $J$ being a discrete set (for instance $J=\{2^k\mid k\in\Z\}$),
together with simplicial maps 
$\simplicialmap_\alpha:K_\alpha\rightarrow K_{\alpha'}$ between
complexes at consecutive scales.
For instance, the Rips filtration can be turned into a tower by restricting 
to a discrete range of scales,
and using the inclusion maps as $\simplicialmap$.
The approximation constructed in this paper will be another example of a tower.

We say that a simplex $\sigma$ is \emph{included} in the tower at scale $\alpha'$,
if $\sigma$ is not the image of 
$\simplicialmap_{\alpha}:K_\alpha\rightarrow K_{\alpha'}$,
where $\alpha$ is the scale preceding $\alpha'$ in the tower.
The \emph{size} of a tower is the number of simplices included over all scales.
If a tower arises from a filtration, its size is simply the size of the 
largest complex in the filtration (or infinite, if no such complex exists).
However, this is not true for in general for
simplicial towers, since simplices can collapse in the tower and the size of the
complex at a given scale may not take into account the collapsed simplices
which were included at earlier scales in the tower.

\subparagraph*{Barcodes and Interleavings}
A collection of vector spaces $(V_\alpha)_{\alpha\in I}$ connected with linear maps
$\linearmap_{\alpha_1,\alpha_2}:V_{\alpha_1}\rightarrow V_{\alpha_2}$ 
is called a \emph{persistence module}, if $\linearmap_{\alpha,\alpha}$ is the 
identity on $V_\alpha$ and 
$\linearmap_{\alpha_2,\alpha_3}\circ\linearmap_{\alpha_1,\alpha_2}
=\linearmap_{\alpha_1,\alpha_3}$ 
for all $\alpha_1\le\alpha_2\le\alpha_3\in I$ for the index set $I$.

We generate persistence modules using the previous concepts. 
Given a simplicial tower $(K_\alpha)_{\alpha\in I}$,
we generate a sequence of chain complexes $(\ch_\ast(K_\alpha))_{\alpha\in I}$.
By functoriality, the simplicial maps $\simplicialmap$ of the tower give rise
to chain maps $\overline{\simplicialmap}$ between these chain complexes.
Using functoriality of homology, we obtain a sequence $(H(K_\alpha))_{\alpha\in I}$
of vector spaces with linear maps $\overline{\simplicialmap}^\ast$, forming
a persistence module. 
The same construction can be applied to filtrations.

Persistence modules admit a decomposition into a collection of 
intervals of the form $[\alpha,\beta]$
(with $\alpha,\beta\in I$), called the \emph{barcode}, subject to certain
tameness conditions.
The barcode of a persistence module characterizes the module uniquely up to 
isomorphism.
If the persistence module is generated by a simplicial complex,
an interval $[\alpha,\beta]$ in the barcode corresponds 
to a homological feature (a ``hole'')
that comes into existence at complex $K_\alpha$ and persists until
it disappears at $K_\beta$. 

Two persistence modules $(V_\alpha)_{\alpha\in I}$ and $(W_\alpha)_{\alpha\in I}$ 
with linear maps $\linearmap_{\cdot,\cdot}$ and $\mu_{\cdot,\cdot}$ are said to be  
\emph{weakly (multiplicatively) $c$-interleaved}
with $c\geq 1$, if there exist linear maps $\gamma_\alpha:V_\alpha\rightarrow 
W_{c\alpha}$ and $\delta_\alpha:W_\alpha\rightarrow V_{c\alpha}$,
called \emph{interleaving maps},
such that the diagram
\begin{equation}
\label{diagram:weak_diag}
\xymatrix{
	& \cdots\ar[r] & V_{\alpha c} \ar[rd]^{\gamma}\ar[rr]^{\lambda} & & 
	V_{\alpha c^3} \ar[r] &\cdots
	\\
	\cdots\ar[r] & W_{\alpha} \ar[rr]^{\mu}\ar[ru]^{\delta} & & 
	W_{\alpha c^2}\ar[r] \ar[ru]^{\delta} &\cdots
	\\
}
\end{equation}
commutes for all $\alpha\in I$, that is, 
$\mu=\gamma \circ \delta $ and $\lambda= \delta\circ \gamma $ 
(we have skipped the subscripts of the maps for readability). 
In such a case, the barcodes of the two modules are $3c$-approximations 
of each other in the sense of~\cite{ch-proximity}. 
We say that two towers are \emph{$c$-approximations} of each other,
if their persistence modules that are $c$-approximations. 

Under more stringent interleaving conditions, the approximation ratio 
can be improved.
Given a totally ordered index set $J$, two persistence modules 
$(V_\alpha)_{\alpha\in J}$ and $(W_\alpha)_{\alpha\in J}$ with
linear maps $\linearmap_{\cdot,\cdot}$ and $\mu_{\cdot,\cdot}$ are said to be  
\emph{strongly (multiplicatively) $c$-interleaved}
with $c\geq 1$, if there exist linear maps $\gamma_\alpha:V_\alpha\rightarrow 
W_{c\alpha}$ and $\delta_\alpha:W_\alpha\rightarrow V_{c\alpha}$,
such that the diagrams
\begin{equation}
\label{diagram:strong_diag}
\xymatrix{
	V_{\frac{\alpha}{c}} \ar[rrr]^{\lambda} \ar[rd]^{\gamma} & & & V_{c\alpha'}    &  & V_{c\alpha} \ar[r]^{\lambda} & V_{c\alpha'} 
	\\
	& W_\alpha \ar[r]^{\mu} & W_{\alpha'} \ar[ru]^{\delta} &                 & W_\alpha \ar[r]^{\mu} \ar[ru]^{\delta} & W_{\alpha'} \ar[ru]^{\delta}
	\\ 
	& V_\alpha \ar[r]^{\lambda} & V_{\alpha'} \ar[rd]^{\gamma} &                 & V_\alpha \ar[r]^{\lambda} \ar[rd]^{\gamma} & V_{\alpha'} \ar[rd]^{\gamma} 
	\\ 
	W_{\frac{\alpha}{c}} \ar[rrr]^{\mu} \ar[ru]^{\delta} & & & 
	W_{c\alpha'}    &  & W_{c\alpha} \ar[r]^{\mu} & W_{c\alpha'}\\
}
\end{equation}
commute for all $\alpha\le \alpha' \in J$.
The barcodes of the two modules are $c$-approximations of each other in the 
sense of~\cite{ch-proximity}.

\section{Grids and cubes}
\label{section:grids}

Let $I:=\{\lambda 2^s\mid s\in\Z\}$ with $\lambda>0$ be a discrete set of scales. 
For a scale $\alpha_s:=\lambda 2^s$, we inductively define a grid $G_s$ on scale 
$\alpha_s$ which is a scaled and translated (shifted) version of the integer 
lattice: for $s=0$, $G_s$ is simply $\lambda\Z^d$, the scaled integer grid.
For $s\geq 0$, we choose an arbitrary $O\in G_s$ and define
\begin{eqnarray}
\label{equation:grids}
G_{s+1} = 2(G_s-O)+O+\frac{\alpha_s}{2}(\pm 1,\ldots,\pm 1)
\end{eqnarray}
where the signs of the components of the last vector are chosen uniformly
at random (and the choice is independent for each $s$).
For $s\leq 0$,
we define
\begin{eqnarray}
\label{equation:grids_2}
G_{s-1} = \frac{1}{2}(G_s-O)+O+\frac{\alpha_{s-1}}{2}(\pm 1,\ldots,\pm 1).
\end{eqnarray}
It is then easy to check that $\ref{equation:grids}$ and $\ref{equation:grids_2}$
are consistent at $s=0$.
A simple instance of the above construction is the sequence of lattices
with $G_s:=\alpha_s\Z^d$ for even $s$, and 
$G_s:=\alpha_s\Z^d + \frac{\alpha_{s-1}}{2}(1,\ldots,1)$ for odd $s$. 

We motivate the shifting next. 
For a finite point set $Q\subset\R^d$ and $x\in Q$, 
the \emph{Voronoi region} $Vor_{Q}(x)\subset \R^d$ is the (closed) set of points 
in $\R^d$ that have $x$ as one of its closest points in $Q$.
If $Q=G_s$, it is easy to see that the Voronoi region of any grid point $x$
is a cube of side length $\alpha_s$ centered at $x$. 
The shifting of the grids ensures that
each $x\in G_s$ lies in the Voronoi region of a unique $y\in G_{s+1}$.
By an elementary calculation, we can show a stronger statement, 
which we use frequently;
for shorter notation, we write $Vor_s(x)$ instead of $Vor_{G_s}(x)$.

\begin{lemma}
\label{lemma:vorcontain}
Let $x\in G_s, y\in G_{s+1}$ such that $x\in Vor_{s+1}(y)$. 
Then, $Vor_{s}(x)\subset Vor_{s+1}(y)$.
\end{lemma}

\begin{proof}
Without loss of generality, we can assume that $\alpha_s=2$ and 
$x$ is the origin, using an appropriate translation and scaling. 
Also, we assume for simplicity that $G_{s+1}=2G_s + (1,\ldots,1)$;
the proof is analogous for any other translation vector.
In that case, it is clear that $y=(1,\ldots,1)$. 
Since $G_s=2\Z^d$, the Voronoi region of $x$ is the set $[-1,1]^d$.
Since $G_{s+1}$ is a translated version of $4\Z^d$,
the Voronoi region of $y$ is the cube $[-1,3]^d$, which covers $[-1,1]^d$.
\end{proof}

\subparagraph*{Cubical complexes}

The integer grid $\Z^d$ naturally defines a \emph{cubical complex},
where each element is an axis-aligned, $k$-dimensional cube with $0\leq k\leq d$.
To define it, let $\square$ denote the set of all integer translates
of faces of the unit cube $[0,1]^d$, considered as a convex polytope in $\R^d$.
We call the elements of $\square$ \emph{faces}. 
Each face has a dimension $k$; the $0$-faces, or \emph{vertices} are exactly
the points in $\Z^d$. 
Moreover, the \emph{facets} of a $k$-face $f$ are the $(k-1)$-faces 
contained in $f$.
We call a pair of facets of $f$ \emph{opposite} if they are disjoint.
Obviously, these concepts carry over to scaled and translated
versions of $\Z^d$, so we can define $\square_s$ as the cubical complex 
defined by $G_s$.

We define a map $g_s: \square_s\rightarrow \square_{s+1}$ as follows:
for vertices, we assign to $x\in G_s$ the (unique) vertex $y\in G_{s+1}$ such that
$x\in Vor_{s+1}(y)$ (cf. Lemma~\ref{lemma:vorcontain}).
For a $k$-face $f$ of $\square_s$ with vertices $(p_1,\ldots,p_{2^k})$ in $G_s$,
we set $g_s(f)$ to be the convex hull of $\{g_s(p_1),\ldots,g_s(p_{2^k})\}$;
the next lemma shows that this is indeed a well-defined map.

\begin{lemma}
	\label{lemma:gcell}
$\{g_s(p_1),\ldots,g_s(p_{2^k})\}$ are the vertices of a face $e$ of $G_{s+1}$.
Moreover, if $e_1,e_2$ are any two opposite facets of $e$, then there exists
a pair of opposite facets $f_1,f_2$ of $f$ such that $g_s(f_1)=e_1$ and
$g_s(f_2)=e_2$.
\end{lemma}

\begin{proof}

\textbf{First claim:}
We prove the first claim by induction on the dimension of faces of $G_s$.
Base case: for vertices, the claim is trivial using Lemma~\ref{lemma:vorcontain}. 
Induction case: let the claim hold true for all $(k-1)$-faces of $G_s$.
We show that the claim holds true for all $k$-faces of $G_s$.

Let $f$ be a $k$-face of $G_s$. 
Let $f_1$ and $f_2$ be opposite facets of $f$, along the $m$-th co-ordinate.
Let the vertices of $f_1$ be $(p_1,\ldots,p_{2^{k-1}})$
and $f_2$ be $(p_{2^{k-1}+1},\ldots,p_{2^{k}})$ taken in the same order, that is,
$p_j$ and $p_{2^{k-1}+j}$ differ in only the $m$-th coordinate for all 
$1\le j\le 2^{k-1}$. 
By definition, all vertices of $f_1$ share the $m$-th coordinate, and we denote
coordinate of these vertices by $z$.
Then, the $m$-th coordinate of all vertices of $f_2$ equals $z+\alpha_s$.
By induction hypothesis, $e_1=g_s(f_1)$ and $e_2=g_s(f_2)$ 
are two faces of $G_{s+1}$. 
We show that the vertices of $e_1\cup e_2$ are vertices of a face $e$ of $G_{s+1}$.

The map $g_s$ acts on each coordinate direction independently. 
Therefore, $g_s(p_j)$ and $g_s(p_{2^{k-1}+j})$ have the same coordinates, 
except possibly the $m$-th coordinate. 
This further implies that $e_2$ is a translate of $e_1$ along the $m$-th
coordinate.

There are two cases: if $e_1$ and $e_2$ share the $m$-th coordinate, then
$e_1=e_2$ and therefore $g_s(f)=e_1=e_2=e$, so the claim follows.
On the other hand, if $e_1$ and $e_2$ 
do not share the $m$-th coordinate: $e_1$'s $m$-th coordinate is $g_s(z)$,
while for $e_2$ it is $g_s(z+\alpha_s)$.
From the structure of $g_s$, we see that $g_s(z)$ and $g_s(z+\alpha_s)$ 
differ by $\alpha_{s+1}$.
It follows that $e_1$ and $e_2$ are two faces of $\square_{s+1}$ which differ 
in only one coordinate by $\alpha_{s+1}$.
So they are opposite facets of a codimension-1 face $e$ of $G_{s+1}$.
Using induction, the claim follows.

\textbf{Second claim:} 
Without loss of generality, assume that $x_1$  is the direction
in which $e_2$ is a translate of $e_1$.
Let $h$ denote the maximal face of $f$ such that $g_s(h)=e_1$.
Clearly, $h\neq f$, since that would imply $g_s(f)=e_1=e$, 
which is a contradiction. 

Suppose $h$ has dimension less than $k-1$.
Let $h'$ be the face of $c$, obtained by translating $h$ along $x_1$.
As in the first claim, it is easy to see that $g_s(h')=e_2$,
from the structure of $g_s$. 
This means that there is a facet $i$ of $f$
containing $h$ and $h'$ such that $g_s(i)=e$. 
Let $i'$ be the opposite facet of $i$ in $f$ and let $x_2$ be the direction
which separates $i$ from $i'$. 
Then, $g_s(i')=e$ otherwise $g_s(f)=e$ does not hold.
Let $h''$ be the face of $f$, obtained by translating $h$ along $x_2$. 
Then, from the structure of $g_s$, $g_s(h'')=e_1$ holds.
The facet of $f$ containing $h$ and $h''$ also maps to $e_1$ under $g_s$.
This is a contradiction to our assumption that $h$ is the highest dimensional
face of $f$ such that $g_s(h)=e_1$. 
See Figure~\ref{figure:gridmap_correct} for a simple illustration.

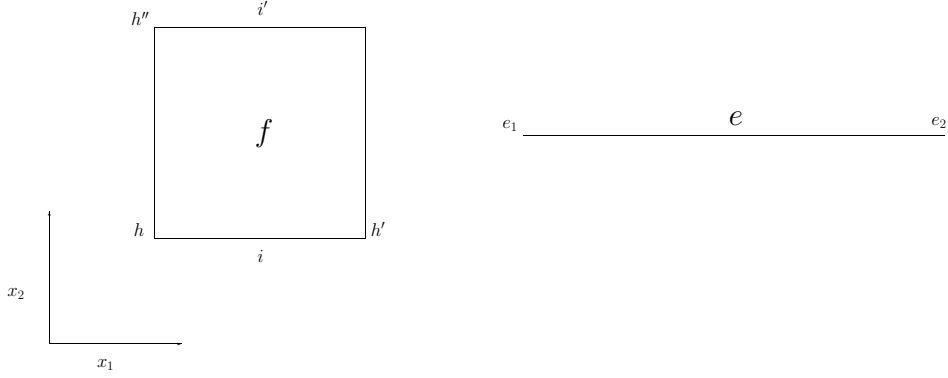
\begin{figure}
	\centering
	\scalebox{0.35}{\setlength{\unitlength}{4144sp}%
\begingroup\makeatletter\ifx\SetFigFont\undefined%
\gdef\SetFigFont#1#2#3#4#5{%
  \reset@font\fontsize{#1}{#2pt}%
  \fontfamily{#3}\fontseries{#4}\fontshape{#5}%
  \selectfont}%
\fi\endgroup%
\begin{picture}(16047,6247)(-2534,-5084)
\thinlines
{\color[rgb]{0,0,0}\put(6301,-1006){\line( 1, 0){7200}}
}%
{\color[rgb]{0,0,0}\put(-1799,-4561){\vector( 1, 0){2250}}
}%
{\color[rgb]{0,0,0}\put(-1799,-4561){\vector( 0, 1){2250}}
}%
{\color[rgb]{0,0,0}\put(  1,-2761){\framebox(3600,3600){}}
}%
\put(-2519,-3751){\makebox(0,0)[lb]{\smash{{\SetFigFont{20}{24.0}{\rmdefault}{\mddefault}{\updefault}{\color[rgb]{0,0,0}$x_2$}%
}}}}
\put(-989,-4966){\makebox(0,0)[lb]{\smash{{\SetFigFont{20}{24.0}{\rmdefault}{\mddefault}{\updefault}{\color[rgb]{0,0,0}$x_1$}%
}}}}
\put(5941,-871){\makebox(0,0)[lb]{\smash{{\SetFigFont{20}{24.0}{\rmdefault}{\mddefault}{\updefault}{\color[rgb]{0,0,0}$e_1$}%
}}}}
\put(13276,-826){\makebox(0,0)[lb]{\smash{{\SetFigFont{20}{24.0}{\rmdefault}{\mddefault}{\updefault}{\color[rgb]{0,0,0}$e_2$}%
}}}}
\put(-359,-2716){\makebox(0,0)[lb]{\smash{{\SetFigFont{20}{24.0}{\rmdefault}{\mddefault}{\updefault}{\color[rgb]{0,0,0}$h$}%
}}}}
\put(-404,884){\makebox(0,0)[lb]{\smash{{\SetFigFont{20}{24.0}{\rmdefault}{\mddefault}{\updefault}{\color[rgb]{0,0,0}$h''$}%
}}}}
\put(1756,-3166){\makebox(0,0)[lb]{\smash{{\SetFigFont{20}{24.0}{\rmdefault}{\mddefault}{\updefault}{\color[rgb]{0,0,0}$i$}%
}}}}
\put(1756,1066){\makebox(0,0)[lb]{\smash{{\SetFigFont{20}{24.0}{\rmdefault}{\mddefault}{\updefault}{\color[rgb]{0,0,0}$i'$}%
}}}}
\put(1711,-1096){\makebox(0,0)[lb]{\smash{{\SetFigFont{34}{40.8}{\rmdefault}{\mddefault}{\updefault}{\color[rgb]{0,0,0}$f$}%
}}}}
\put(3691,-2716){\makebox(0,0)[lb]{\smash{{\SetFigFont{20}{24.0}{\rmdefault}{\mddefault}{\updefault}{\color[rgb]{0,0,0}$h'$}%
}}}}
\put(9811,-826){\makebox(0,0)[lb]{\smash{{\SetFigFont{34}{40.8}{\rmdefault}{\mddefault}{\updefault}{\color[rgb]{0,0,0}$e$}%
}}}}
\end{picture}%
}
	\caption[Grid map]{
		The face $f$ is a square, for which $g(f)=e$ is a line segment.
		The horizontal and vertical directions are $x_1$ and $x_2$ respectively.
		The rest of the labels are self-explanatory in relation to the text.
		}
		\label{figure:gridmap_correct}
		\end{figure}
		
		Therefore, the only possibility is that $h$ is a facet $f_1$ of $f$ such 
		that $g_s(f_1)=e_1$. Let $f_2$ be the opposite facet of $f_1$. 
		From the structure of $g_s$, it is easy to see that $g_s(f_2)=e_2$.
		The claim follows.
		\end{proof}
			
\subparagraph*{Barycentric subdivision}

A \emph{flag} in $\square_s$ is a set of faces $\{f_0,\ldots,f_k\}$ of $\square_S$
such that $f_0\subseteq \ldots\subseteq f_k$.
The \emph{barycentric subdivision} $sd_s$ of $\square_s$ is the 
(infinite) simplicial complex
whose simplices are the flags of $\square_s$; in particular, the $0$-simplices
of $sd_s$ are the faces of $\square_s$. 
An equivalent geometric description of $sd_s$ can be obtained by defining the 
$0$-simplices as the barycenters of the faces in $sd_s$, and introducing a 
$k$-simplex between $(k+1)$ barycenters if the corresponding faces form a flag. 
It is easy to see that $sd_s$ is a flag complex.
Given a face $f$ in $\square_s$, we write $sd(f)$ for the subcomplex of $sd_s$
consisting of all flags that are formed only by faces contained in $f$.

\section{Approximation scheme}
\label{section:scheme}

We define our approximation complex at scale $\alpha_s$
as a finite subcomplex of $sd_s$. 
To simplify the subsequent analysis,
we define the approximation in a slightly generalized form.

\subparagraph*{Barycentric spans}

For a fixed $s$, let $V$ denote a non-empty subset of $G_s$.
We say that a face $f\in\square_s$ is \emph{spanned} by $V$
if $f\cap V$ is non-empty and not contained in any facet of $f$.
Trivially, the vertices of $\square_s$ spanned by $V$ are precisely 
the points in $V$.
We point out that the set of spanned faces is \emph{not} closed
under taking sub-faces; for instance, if $V$ consists of two antipodal
points of a $d$-cube, the only faces spanned by $V$ are the $d$-cube
and the two vertices.

The \emph{barycentric span} of $V$ is the subcomplex of $sd_s$ defined
by all flags $\{f_0,\ldots,f_k\}$ such that all $f_i$ are spanned by $V$. 
This is indeed a subcomplex of $sd_s$ because it is closed
under taking subsets. 
Moreover, for a face $f\in\square_k$,
we define the \emph{$f$-local barycentric span} of $V$ as the set of all flags
$\{f_0,\ldots,f_k\}$ in the barycentric span such that $f_i\subseteq f$ for all $i$.
This is a subcomplex both of $sd(f)$ and of the barycentric span of $V$
and is a flag complex.

\begin{lemma}
\label{lemma:uxcomplex}
For each face $f$, the $f$-local barycentric span of $V$ is either empty or acyclic.
\end{lemma}

\begin{proof}	
We assume that the $f$-local barycentric span of $V$ is not empty. 
Hence, $f$ contains a unique active face $e$ of maximal dimension 
that is spanned by $V$.
A simplex in a simplicial complex $K$ is called \emph{maximal} if no other simplex 
in $K$ contains it. 
It is known that if a simplicial complex $K$ contains a $0$-simplex $\sigma$
that lies in every maximal simplex, than $K$ is acyclic 
(in this case, $K$ is called \emph{star-shaped}).
In our situation, the $0$-simplex $e$ belongs to every maximal simplex 
in the $f$-local barycentric span,
because every simplex not containing $e$ is a flag that can be extended 
by adding $e$ to its end.
\end{proof}
	
Furthermore, if $W\subseteq V$, it is easy to see that faces spanned by $W$
are also spanned by $V$. 
Consequently, the barycentric span of $W$
is a subcomplex of the barycentric span of $V$.

\subparagraph{Approximation complex}
We denote by $P\subset \R^d$ a finite set of points.
For each point $p\in P$, we let $a_s(p)$ denote the grid
point in $G_s$ that is closest to $p$ (we assume for simplicity
that this closest point is unique). 
We define the \emph{active vertices of $G_s$}, $V_s$, 
as $a_s(P)$, that is, the set of grid points that are closest to some point in $P$.
The next statement is a direct application of the triangle inequality; 
let $\dmn$ denote the diameter in the $\lin$-norm. 
\begin{lemma}
\label{lemma:iripscell}
Let $Q\subseteq P$ be such that $\dmn(Q)\le \alpha_s$. 
Then, the set $a_s(Q)$ is contained in a face of $\square_s$. 
Equivalently, for a simplex $\sigma=(p_0,\ldots,p_k)\in\rin_{\alpha_s}$ on $P$, 
the set of active vertices 
$\{a_{s}(p_0),\ldots,a_{s}(p_k)\}$ is contained in a face of $\square_{s}$.
\end{lemma}

\begin{proof}
We prove the claim by contradiction. 
Assume that the set of active vertices $a_s(Q)$ is not 
contained in a face of $\square_s$. 
Then, there exists $x,y\in Q$ such 
$a_s(x)$, $a_s(y)$ are not in a common face of $\square_s$. 
By the definition of the grid $G_s$, the grid points $a_s(x)$, $a_s(y)$ 
therefore have $L_\infty$-distance at least $2\alpha_s$.
Moreover, $x$ has $L_\infty$-distance less than $\alpha_s/2$ from $a_s(x)$, 
and the same is true for $y$ and $a_s(y)$.
By triangle inequality, the $L_\infty$-distance of $x$ and $y$ 
is more than $\alpha_s$, a contradiction.
\end{proof}

Vice versa, we define a map $b_s:V_s\rightarrow P$ by mapping an active vertex 
to its closest point in $P$ 
(again, assuming for simplicity that the assignment is unique).
The map $b_s$ is a section of $a_s$, that is, $a_s\circ b_s$ 
is the identity on $V_s$.

Recall that the map $g_s:\square_s\rightarrow\square_{s+1}$ from 
Section~\ref{section:grids} maps grid points of $G_s$ to grid points of $G_{s+1}$. 
With Lemma~\ref{lemma:vorcontain}, it follows at once:

\begin{lemma}
\label{lemma:gcompose}
For all $x\in V_s$, $g_s(x)=(a_{s+1}\circ b_s)(x)$.
\end{lemma}

We now define our approximation tower: for scale $\alpha_s$, we define 
$\ux_{\alpha_s}$ as the barycentric span of the active vertices $V_s\subset G_s$.
See Figure~\ref{fig:barycpx} for an illustration. 
To simplify notations, we call the faces of $\square_s$
spanned by $V_s$ \emph{active faces}, and simplices of 
$\ux_{\alpha_s}$ \emph{active flags}.

\begin{figure}
\centering
	\includegraphics[width=0.4\textwidth]{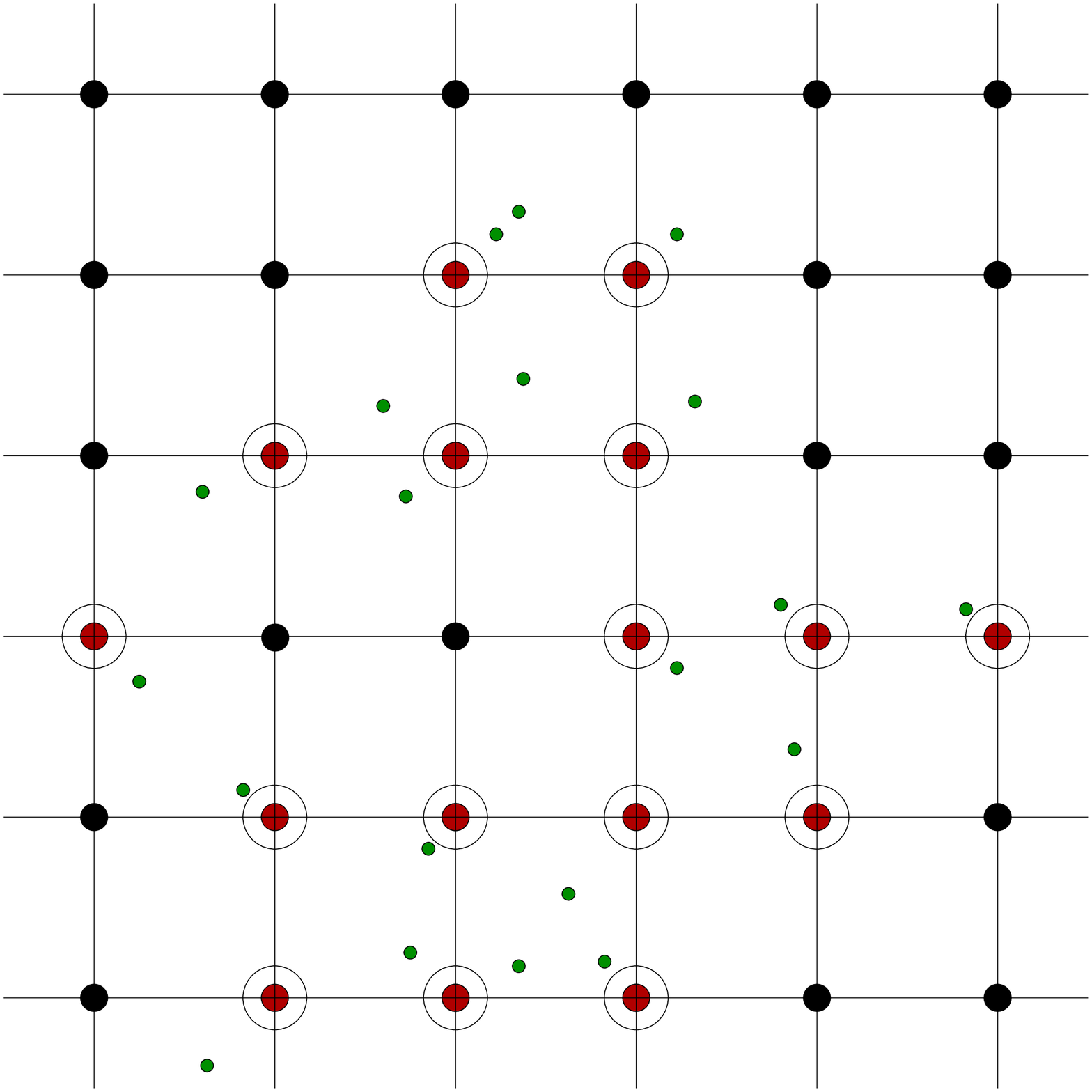}\hspace{5mm}%
	\includegraphics[width=0.4\textwidth]{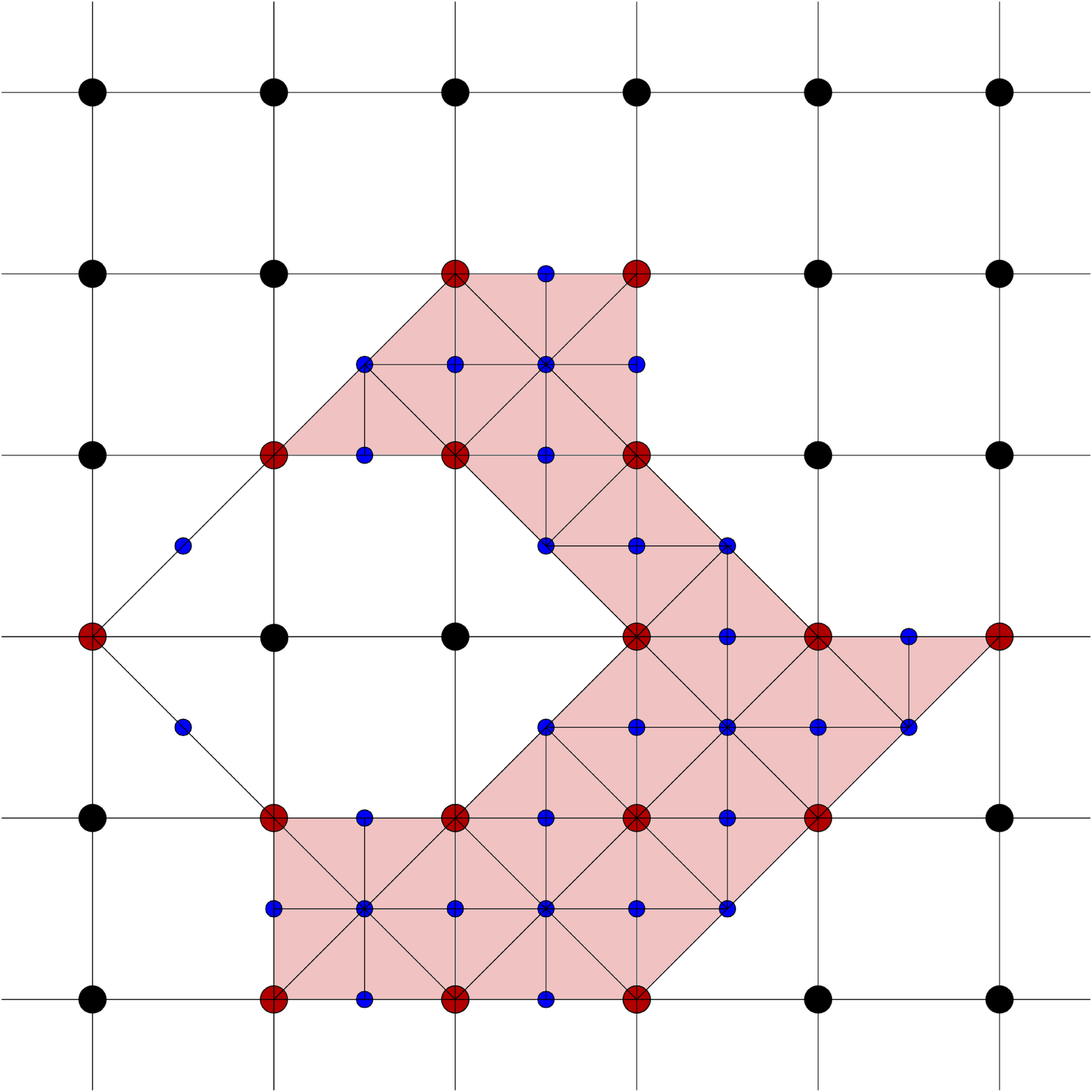}
\caption{
The left figure shows a two-dimensional grid, along with its cubical complex.
The green points (small dots) denote the
points in $P$ and the red vertices (encircled) are the active vertices.
The figure on the right shows the generated simplicial complex. 
The blue vertices (small dots) are the barycenters of the active faces.
}
\label{fig:barycpx}
\end{figure}

To complete the construction, we need to define simplicial maps 
$\ux_{\alpha_s}\rightarrow\ux_{\alpha_{s+1}}$.
We show that such maps are induced by $g_s$. 

\begin{lemma}
\label{lemma:activeimage}
Let $f$ be an active face of $\square_s$. 
Then, $g_s(f)$ is an active face of $\square_{s+1}$.
\end{lemma}

\begin{proof}
From Lemma~\ref{lemma:gcell}, $e:=g_s(f)$ is a face of $G_{s+1}$. 
If $e$ is a vertex, it is active, because $f$ contains
at least one active vertex $v$, and $g_s(v)=e$ in this case.
If $e$ is not a vertex, we assume for a contradiction that it is not active.
Then, it contains a facet $e_1$ that contains all active vertices in $e$.
Let $e_2$ denote the opposite facet. 
By Lemma~\ref{lemma:gcell}, $f$ contains
opposite facets $f_1$, $f_2$ such that $g_s(f_1)=e_1$ and $g_s(f_2)=e_2$.
Since $f$ is active, both $f_1$ and $f_2$ contain active vertices,
in particular, $f_2$ contains an active vertex $v$. 
But then, the active vertex $g_s(v)$ must lie in $e_2$, contracting the fact 
that $e_1$ contains all active vertices of $e$.
\end{proof}

Recall that a simplex $\sigma\in \ux_{\alpha_s}$ is a flag 
$f_0\subseteq \ldots \subseteq f_k$ of active faces in $\square_s$.
We set $\tg(\sigma)$ as the flag $g(f_0)\subseteq \ldots \subseteq g(f_k)$, 
which consists of active
faces in $\square_{s+1}$ by Lemma~\ref{lemma:activeimage}, 
and hence is a simplex in $\ux_{\alpha_{s+1}}$.
It follows that $\tg:\ux_{\alpha_s}\rightarrow\ux_{\alpha_{s+1}}$ 
is a simplicial map.
This finishes our construction of the simplicial tower
$
(\ux_{\lambda 2^s})_{s\in\Z},
$
with simplicial maps $\tg:\ux_{\lambda 2^s}\rightarrow\ux_{\lambda 2^{s+1}}$.

\subsection{Interleaving}
\label{subsection:intlv}
To relate our tower with the $\lin$-Rips filtration,
we start by defining two acyclic carriers. 
We write $\alpha:=\alpha_s=\lambda 2^s$ to simplify notations.

\begin{itemize}
\item $C_1:\rin_{\alpha} \rightarrow \ux_{\alpha}$: let $\sigma=(p_0,\ldots,p_k)$
be any simplex of $\rin_\alpha$. 
We set $C_1(\sigma)$ as the barycentric span of $U:=\{a_s(p_0),\ldots,a_s(p_k)\}$, 
which is a subcomplex of $\ux_{\alpha}$.
$U$ lies in a face $f$ of $\square_s$ by Lemma~\ref{lemma:iripscell}
hence $C_1(\sigma)$ is also the $f$-local barycentric span of $U$.
Using Lemma~\ref{lemma:uxcomplex}, $C_1(\sigma)$ is acyclic.

\item $C_2:\ux_{\alpha}\rightarrow \rin_{2\alpha}$: let $\sigma$ be any 
flag $e_0\subseteq\ldots \subseteq e_k$ of $\ux_\alpha$. 
Let $\{q_0,\ldots,q_m\}$ be the set of active vertices of $e_k$.
We set $C_2(\sigma):=\{b_s(q_0),\ldots,b_s(q_m)\}$. 
With a simple triangle inequality, we see that $C_2(\sigma)$ 
is a simplex in $\rin_{2\alpha}$, hence it is acyclic.
\end{itemize}
Using the Acyclic Carrier Theorem (Theorem~\ref{theorem:acyclic_carrier}),
there exist chain maps
$c_1:\ch_\ast(\rin_\alpha)\rightarrow\ch_\ast(\ux_{\alpha})$ and
$c_2:\ch_\ast(\ux_\alpha)\rightarrow\ch_\ast(\rin_{2\alpha})$, 
which are carried by $C_1$ and $C_2$, respectively. 
Aggregating the chain maps, we have the following diagram:
\begin{eqnarray}
\label{diagram:inf_intlv}
\xymatrix{
	& \cdots\ar[r] & \ch_\ast(\rin_{2\alpha}) \ar[d]^{c_1}\ar[r]^{inc} & 
	\ch_\ast(\rin_{4\alpha}) \ar[r] &\cdots
	\\
	\cdots\ar[r] & \ch_\ast(\ux_{\alpha}) \ar[r]^{\tg}\ar[ru]^{c_2} & 
	\ch_\ast(\ux_{2\alpha })\ar[r] \ar[ru]^{c_2} &\cdots
	\\
}
\end{eqnarray}
where $inc$ corresponds to the inclusion chain map
and $\tg$ denotes the chain map for the corresponding
simplicial maps (we removed indices for readability). 
The chain complexes give rise to a diagram of the corresponding homology groups,
connected by the induced linear maps $c_1^\ast,c_2^\ast,inc^\ast,\tg^\ast$.

\begin{lemma}
	\label{lemma:inf_intlv}
$inc^\ast=c_2^\ast\circ c_1^\ast$ and $\tg^\ast=c_1^\ast\circ c_2^\ast$.
In particular, the persistence modules $\left(H(\ux_{2^s})\right)_{s\in\Z}$ and 
$\left(H(\rin_\alpha)\right)_{\alpha\ge 0}$ are weakly $2$-interleaved.
\end{lemma}

\begin{proof}
To prove the claim, we consider both triangles separately.
We show that the chain maps $\tg$ and $c_1\circ c_2$ 
are carried by a common acyclic carrier. 
Then we show the same statement for $inc$ and $c_2\circ c_1$. 
The claim then follows from the Acyclic Carrier Theorem.

\begin{itemize}
\item \emph{Lower triangle}: The map
$C_1\circ C_2:\ux_\alpha\rightarrow \ux_{2\alpha}$ is an acyclic carrier, 
because $C_2(\sigma)$ is a simplex for any simplex $\sigma\in \us_\alpha$. 
Clearly, $C_1\circ C_2$ carries the map $c_1\circ c_2$. 
We show that it also carries $\tg$.

Let $\sigma$ be a flag $f_0\subseteq\ldots\subseteq f_k$ in $\ux_{\alpha}$
and let $V(f_i)$ denote the active vertices of $f_i$.
Then, $C_1\circ C_2(\sigma)$ is the barycentric span of 
\[
U:=\{a_{s+1}\circ b_s(q)\mid q\in V(f_k)\}=\{g_{s}(q)\mid q\in V(f_k)\}
\]
(Lemma~\ref{lemma:gcompose}).
On the other hand, $V(f_i)\subseteq V(f_k)$ and hence $g(V(f_i))\subseteq U$.
Then, $g(f_i)$ is spanned by $U$: indeed, since $f_i$ is active, 
$g(f_i)$ is active and hence spanned by all active vertices,
and it remains spanned if we remove all active vertices not in $U$, 
since they are not contained in $f_i$.
It follows that the flag $g(f_0)\subseteq\ldots\subseteq g(f_k)$, 
which is equal to $\tg(\sigma)$, is in the barycentric span of $U$.

\item \emph{Upper triangle}:
We define an acyclic carrier $D:\rin_{2\alpha}\rightarrow\rin_{4\alpha}$ which
carries both $inc$ and $c_2\circ c_1$.
Let $\sigma=(p_0,\ldots,p_k)\in \rin_{2\alpha}$ be a simplex. 
The active vertices $U:=\{a(p_0),\ldots,a(p_k)\}\subset G_{s+1}$ 
lie in a face $f$ of $G_{2\alpha}$, using Lemma~\ref{lemma:iripscell}.
We can assume that $f$ is active, as otherwise, we pass to a facet
of $f$ that contains $U$.
We set $D(\sigma)$ as the simplex on the subset of points in $P$
whose closest grid point in $G_{s+1}$ lies in $U$.
Using a simple application of triangle inequalities, 
$D(\sigma)\in \rin_{4\alpha}$, so $D$ is an acyclic carrier.
The $0$-simplices of $\sigma$ are a subset of $D(\sigma)$, 
so $D$ carries the map $inc$. 
We next show that $D$ carries $c_2\circ c_1$. 

Let $\delta$ be a simplex in $\ux_{2\alpha}$ for which the chain
$c_1(\sigma)$ takes a non-zero value. 
Since $c_1(\sigma)$ is carried by
$C_1(\sigma)$, $\delta\in C_1(\sigma)$ which is the barycentric span of $V(f)$. 
Furthermore, for any $\tau\in C_1(\sigma)$, $C_2(\tau)$ is of the form
$\{b(q_0),\ldots,b(q_m)\}$ with $\{q_0,\ldots,q_m\}\in V(f)$.
It follows that $C_2(\tau)\subseteq D(\sigma)$. 
In particular, since $c_2$ is carried by $C_2$, 
$c_2(c_1(\sigma))\subseteq D(\sigma)$ as well.
\end{itemize}
\end{proof}

\subsection{Scale balancing}
We improve the approximation factor with a simple modification 
that we first explain in general.
Let $(A_{\lambda\gamma^k})_{k\in\Z}$ and $(B_{\lambda\gamma^k})_{k\in\Z}$ be two 
simplicial towers with simplicial maps $f_3$ and $f_4$ respectively, 
with $\lambda,\gamma>0$.
Assume that there exist interleaving linear maps $f_1^\ast,f_2^\ast$ 
such that the diagram
\begin{eqnarray}
\label{diagram:rename_1}
\xymatrix{
& \cdots\ar[r] & H(B_{\alpha\gamma}) \ar[d]^{f_2^\ast}\ar[r]^{f_4^\ast} & 
H(B_{\alpha\gamma^2}) \ar[r] &\cdots
\\
\cdots\ar[r] & H(A_{\alpha}) \ar[r]^{f_3^\ast}\ar[ru]^{f_1^\ast} & 
H(A_{\alpha\gamma})\ar[r] \ar[ru]^{f_1^\ast} &\cdots
\\
}
\end{eqnarray}
commutes for all scales $\alpha=\lambda\gamma^k$, 
which implies that the persistence modules are weakly $\gamma$-interleaved.
Defining another tower $(A'_{\lambda\sqrt{\gamma}\gamma^k})_{k\in\Z}$
with $A'_\alpha:=A_{\alpha/\sqrt{\gamma}}$, we obtain a diagram
\begin{eqnarray}
\label{diagram:rename_2}
\xymatrix{
	& \cdots\ar[r] & H(B_{\alpha \gamma}) \ar[rd]^{f_2^\ast}\ar[rr]^{f_4^\ast} & & 
	H(B_{\alpha \gamma^2}) \ar[r] &\cdots
	\\
	\cdots\ar[r] & H(A'_{\alpha\sqrt{\gamma}}) \ar[rr]^{f_3^\ast}\ar[ru]^{f_1^\ast} & & 
	H(A'_{\alpha \gamma\sqrt{\gamma}})\ar[r] \ar[ru]^{f_1^\ast} &\cdots
	\\
}
\end{eqnarray}
which implies that the persistence modules are weakly $\sqrt{\gamma}$-interleaved.
Therefore, scale balancing improves the 
interleaving ratio by only scaling the persistence module.

In our context, we can improve the weak $2$-interleaving of 
$(H(\ux_{2^k\alpha}))_{k\in\Z}$ and $(H(\rin_\alpha))_{\alpha\ge 0}$ 
to a weak $\sqrt{2}$-interleaving. 
Using the proximity results for persistence modules~\cite{ch-proximity}, 
\begin{theorem}
	\label{theorem:irips_ratio}
The persistence module $\big(H(\ux_{2^k/\sqrt{2}})\big)_{k\in\Z}$  is a 
$3\sqrt{2}$- approximation of 
the $L_\infty$-Rips persistence 
module $\big(H(\rin_\alpha)\big)_{\alpha\ge 0}$.
\end{theorem}

For any pair of points $p,p'\in \R^d$, it holds that 
$\|p-p'\|_2\le \|p-p'\|_\infty\le \sqrt{d}\,\|p-p'\|_2$
which implies that the $L_2$- and the $L_\infty$-Rips complexes are strongly
$\sqrt{d}$-interleaved.
The scale balancing technique also works for strongly interleaved 
persistence modules and yields
\begin{lemma}
	\label{lemma:rips_relation}
$(H(\ri_{\alpha/d^{0.25}}))_{\alpha\ge 0}$ is strongly 
$d^{0.25}$-interleaved with $(H(\rin_\alpha))_{\alpha\ge 0}$.
\end{lemma}
Using Theorem~\ref{theorem:irips_ratio}, Lemma~\ref{lemma:rips_relation}
and the fact that interleavings satisfy the triangle 
inequality~\cite[Theorem~3.3]{bs-categorization}, we see that 
$(H(\ux_{2^k/\sqrt{2}}))_{k\in\Z}$ is weakly
$\sqrt{2}d^{0.25}$-interleaved with the scaled Rips module
$(H(\ri_{\alpha/d^{0.25}}))_{\alpha\ge 0}$.
We can remove the scaling in the Rips filtration simply by multiplying both
sides with $d^{0.25}$ and obtain our final approximation result.
\begin{theorem}
\label{theorem:rips_ratio}
The persistence module 
$\big(H(\ux_{\frac{2^k\sqrt[4]{d}}{\sqrt{2}}})\big)_{k\in\Z}$ is a 
$3\sqrt{2}d^{0.25}$-approximation of the Euclidean Rips persistence module 
$\big(H_\ast(\ri_{\alpha})\big)_{\alpha\ge 0}$.
\end{theorem}

\section{Size and computation}
\label{section:computational}

Set $n:=|P|$ and let $CP(P)$ denote the closest pair distance of $P$. 
At scale $\alpha_0:=\frac{CP(P)}{3d}$ and lower, no $d$-cube of the cubical complex 
contains more than one active vertex,
so the approximation complex consists of $n$ isolated $0$-simplices.
At scale $\alpha_m:=diam(P)$ and higher, points of $P$ map to active vertices
of a common face by Lemma~\ref{lemma:iripscell}, so 
the generated complex is acyclic using Lemma~\ref{lemma:uxcomplex}. 
We inspect the range of scales $[\alpha_0,\alpha_m]$ to construct the tower, 
since the barcode is explicitly known for scales outside this range.
The total number of scales is $\ceil{\log_2 \alpha_m/\alpha_0}=
\ceil{\log_2\Delta +\log_2 3d}=O(\log\Delta+\log d)$.

\subsection{Size of the tower}
\label{subsec:size}
Recall that the size of a tower is the number of simplices that do not have a 
preimage.
We start by considering the case of $0$-simplices.

\begin{lemma}
\label{lemma:vtx_inclusion}
The number of $0$-simplices included in the tower is at most $n2^{O(d)}$.
\end{lemma}

\begin{proof}
Recall that the $0$-simplices of $\ux_\alpha$ are the active faces of the cubical 
complex $\square$ at the same scale, and that the simplicial map $\tg$ 
restricted to the $0$-simplices corresponds to the cubical map $g$.

We first consider the active vertices: at scale $\alpha_0$, there are $n$ 
inclusions of $0$-simplices in the tower, due to $n$ active vertices. 
By Lemma~\ref{lemma:vorcontain}, $g$ is surjective
on the active vertices of $\square$ (for any scale). 
Hence, no further active vertices are added to the tower.

It remains to count the active faces of dimension $\geq 1$ without preimage.
We will use a charging argument, charging the existence of such an active face
to one of the points in $P$, charging each point at most $3^{d}$ times.
For that, we fix an arbitrary total order $\prec$ on $P$. 
Each active vertex on any scale has a non-empty subset of $P$ in its Voronoi 
region; we call the maximal such point with respect to $\prec$
the \emph{representative}. 
For an active face $f$ without preimage under $g$,
$f$ has at least two incident active vertices, with distinct representatives.
We charge the inclusion of $f$ to the minimal representative 
among the incident active vertices.

Let $M$ be the number of incident faces of a vertex in the 
cubical grid $\square$ (for any scale).
As one can easily see with combinatorial arguments, $M=3^d=2^{O(d)}$.
Assume for a contradiction that a a point $p\in P$ is charged more than $M$ times.
Whenever any face $f_i$ is charged to $p$, there is an active vertex
$v_i$ whose representative is $p$.
We enumerate these as the set of active vertices 
$\{v_0,\ldots,v_m\}$ on the scales $\alpha_0,\ldots,\alpha_m$ such that
$p$ is the representative of $v_i$ on scale $\alpha_i$.
Naturally, for any $v_i$ and $v_j$, there is a canonical isomorphism between
the $M$ faces incident to $v_i$ and the $M$ faces incident to $v_j$.

Since we assumed that $p$ is charged for $>M$ active faces,
by pigeonhole principle, there must be two vertices $v_i$ and $v_j$ with $i<j$
such that a pair of isomorphic incident faces are charged for $v_i$ and for $v_j$.
There is a sequence of isomorphic faces $f_{i}, f_{i+1},\ldots,f_j$ corresponding
to $v_i,v_{i+1}\ldots,v_j$, respectively, 
such that $p$ is charged for $f_i$ and $f_j$.
Since $f_i$ and $f_j$ have both no preimage, there must be some $f_\ell$ with 
$i<\ell<j$ such that $f_\ell$ is not active. 
That means, however, that the Voronoi region of $v_\ell$
is the union of at least two Voronoi regions of vertices incident to $v_i$.
In that case, because we choose the representative by minimizing over
the maximal representatives, so $p$ is not the representative of $v_\ell$, 
and hence, not of $v_j$. 
This is a contradiction to our claim, so $M$ can not 
be charged more than $M$ times.
\end{proof}

The next lemma follows from a simple combinatorial counting argument 
for the number of flags in a $d$-dimensional cube.

\begin{lemma}
\label{lemma:starsize}
Each $0$-simplex of $\ux_\alpha$ has at most $2^{O(d\log k)}$
incident $k$-simplices.
\end{lemma}

\begin{proof}
A $0$-simplex in $\ux\alpha$ corresponds to an active face $f$ 
in a cubical complex $\square$.
An incident simplex corresponds to an active flag of $\square$ involving $f$.
Let $c$ be a $d$-cube of $\square$ that contains $f$.
We simply count the number of flags of length $(k+1)$ contained in $c$ 
(regardless of whether they contain $f$ or not) and show that the number is 
$2^{O(d\log k)}$. 
Since $f$ is contained in at most $2^d$ $d$-cubes, the bound follows.

To count the number of flags containing $c$, we us similar ideas 
as in~\cite{brunker-dgds}: first fix a vertex
$v$ of $c$ and count the flags of the form $v\subseteq\ldots\subseteq c$. 
Every $\ell$-face in $c$ incident to $v$ corresponds to a subset of $\ell$ 
coordinate indices, in the sense that the coordinates not chosen are fixed 
to the coordinates of $v$ for the face. 
With this correspondence, it is not hard to see that a flag from $v$ to $c$ of 
length $k+1$ corresponds to an ordered $k$-partition of $\{1,\ldots,d\}$.
The number of such partitions is known as $k!$ times the quantity
$\left\{\begin{array}{c}d\\k\end{array}\right\}$, which is the Stirling number
of second kind, and is upper bounded by $2^{O(d\log k)}$~\cite{ckr-polynomial}. 
Since $c$ has $2^d$ vertices, the total number
of flags of the form $v\subseteq\ldots\subseteq c$ with any vertex $v$ is hence 
$2^d k! 2^{O(d\log k)}=2^{O(d\log k)}$.

For flags which do not start with a vertex and do not end with $c$, we can simply 
extend by adding a vertex and/or the $d$-cube and obtain flags 
of length $k+2$ or $k+3$. 
The same argument as above shows again that the number of such flags is
bounded by $2^{O(d\log k)}$ which proves the claim.
\end{proof}

\begin{theorem}
\label{theorem:towersize}
The $k$-skeleton of the tower has size at most $n2^{O(d\log k)}$.
\end{theorem}

\begin{proof}  
Let $\sigma=f_0\subseteq\ldots \subseteq f_k$ be a flag included 
at some scale $\alpha$.
The crucial insight is that this can only happen if 
at least one face $f_i$ in the flag
is included in the tower at the same scale. 
Indeed, if each $f_i$ has a preimage $e_i$
on the previous scale, then $e_0\subseteq \ldots\subseteq e_k$ 
is a flag on the previous scale which maps to $\sigma$ under $\tg$. 

We charge the inclusion of the flag to the inclusion of $f_i$. 
By Lemma~\ref{lemma:starsize}, the $0$-simplex $f_i$ of $\ux$ is charged at most 
$\sum_{i=1}^{k}2^{O(d\log i)}=2^{O(d\log k)}$ times in this way, 
and by Lemma~\ref{lemma:vtx_inclusion}, there are at most $n2^{O(d)}$ $0$-simplices
that can be charged. 
\end{proof}

\subsection{Computing the tower}
\label{subsection:computing}

Recall from the construction of the grids that $G_{s+1}$ is built from $G_s$
using an arbitrary translation vector $(\pm 1,\ldots,\pm 1)\in\Z^d$.
In our algorithm, we pick the components of this translation vector 
uniformly at random, and independently for each scale.

Recall the cubical map $g_s:\square_s\rightarrow\square_{s+1}$ 
from Section~\ref{section:grids}.
For a fixed $s$, we denote by $g^{(j)}:\square_s\rightarrow\square_{s+j}$ the $j$-fold composition of $g$,
that is $g^{(j)}=g_{s+j-1}\circ g_{s+j-2}\circ\ldots\circ g_s$.

\begin{lemma}
\label{lemma:g_survival}
For a $k$-face $f$ of $\square_s$, let $Y$ be the minimal integer $j$ 
such that $g^{(j)}(f)$ is a vertex.
Then $E[Y]\leq 3\log k$.
\end{lemma}

\begin{proof}
Without loss of generality, assume that the grid under consideration is $\Z^d$
and $f$ is the $k$-face spanned by the vertices 
$\{\underbrace{\pm 1,\ldots,\pm 1}_{k},0,\ldots,0\}$.
The proof for the general case is analogous.

Let $y_1\in\{-1,1\}$ denote the randomly chosen shift of the first coordinate. 
If $y_1=1$, the grid $G'$ on the next scale has a grid point 
with $x_1$-coordinate $1/2$. 
Clearly, the closest grid point in $G'$ to the origin is of the form 
$(+1/2,\pm 1/2,\ldots,\pm 1/2)$, and thus, 
this point is also closest to $(1,0,0,\ldots,0)$. 
The same is true for any point $(0,\ast,\ldots,\ast)$ and its corresponding point 
$(1,\ast,\ldots,\ast)$ on the opposite facet.
Hence, for $y_1=1$, $g(f)$ is a face where all points 
have the same $x_1$-coordinate.

Contrarily, if $y_1=-1$, the origin is mapped to some point 
$(-1/2,\pm 1/2,\ldots,\pm 1/2)$ and $(1,0,\ldots,0)$ is mapped to 
$(3/2,\pm 1/2,\ldots,\pm 1/2)$, as one can directly verify.
Hence, in this case, in $g(f)$, points do not all have the same $x_1$ coordinate.

We say that the $x_1$-coordinate \emph{collapses} in the first case and 
does not collapse in the second.
Because the shift is chosen uniformly at random for each scale, 
the probability that $x_1$ did not collapse after $j$ iterations
is $1/2^{j}$.

$f$ spans $k$ coordinate directions, so it must
collapse along each such direction to contract to a vertex.
Once a coordinate collapses, it stays collapsed at all higher scales.
As the shift is independent for each coordinate
direction, the probability of a collapse is the same 
along all coordinate directions that $f$ spans.
Using union bound, the probability that $g^j(f)$ has not collapsed 
to a vertex is at most $k/2^j$. 
With $Y$ as in the statement of the lemma
$P(Y\ge j)\le k/2^j$. 
Hence, 

\begin{align*}
E[Y]=\sum_{j=1}^{\infty} j P(Y=j) &= \sum_{j=1}^{\infty} P(Y\ge j) 
\le  \log k + \sum_{c=1}^{\infty}\sum_{j=c\log k}^{(c+1)\log k} P(Y\ge j)\\
& \le \log k + \sum_{c=1}^{\infty}\sum_{j=c\log k}^{(c+1)\log k}
P(Y\ge c\log k)
\le \log k+ \sum_{c=1}^{\infty}\log k\frac{k}{2^{c \log k}} \\
& \le \log k+ \log k\sum_{c=0}^{\infty} \frac{1}{k^{c}}
\le \log k + 2\log k \le 3 \log k. 
\end{align*}

\end{proof}

As a consequence of the lemma, the expected ``lifetime'' of $k$-simplices 
in our tower with $k>0$ is rather short: Given a flag 
$e_0\subseteq\ldots\subseteq e_\ell$, the face $e_\ell$ will be mapped to a vertex
after $O(\log d)$ steps, and so will be all its sub-faces, 
turning the flag into a vertex.
It follows that summing up the total number of $k$-simplices with $k>0$ over all 
$\ux_\alpha$ yields an upper bound of $n2^{O(d\log k)}$ as well.

\subparagraph*{Algorithm description}
\label{para:algo1}
We first specify what it means to ``compute'' the tower. 
We make use of the fact that
a simplicial map between simplicial complexes can be written 
as a composition of simplex inclusions
and contractions of $0$-simplices~\cite{desu-gic,kstwr}. 
That is, when passing from a scale $\alpha_s$ to $\alpha_{s+1}$, 
it suffices to specify which pairs of $0$-simplices in $\ux_{\alpha_s}$ 
are mapped to the same image under $\tg$ and which 
simplices in $\ux_{\alpha_{s+1}}$ are included. 

The input is a set of $n$ points $P\subset \R^d$.
The output is a list of \emph{events}, where each event is of 
one of the three following types:
a \emph{scale event} defines a real value $\alpha$ and 
signals that all upcoming events 
happen at scale $\alpha$ (until the next scale event). 
An \emph{inclusion event} introduces a new simplex, 
specified by the list of $0$-simplices
on its boundary (we assume that every $0$-simplex is identified by an integer). 
A \emph{contraction event} is a pair of $0$-simplices $(i,j)$ and signifies that 
$i$ and $j$ are identified as the same from that scale.

In a first step, we calculate the range of scales that we are interested in.
We compute a $2$-approximation of $diam(P)$ 
by taking any point $p\in P$ and calculating $max_{q\in P}\|p-q\|$.
Then we compute $CP(P)$ using 
a randomized algorithm in $n2^{O(d)}$ expected time~\cite{cp_khas}.

Next, we proceed scale-by-scale and construct the list of events accordingly.
On the lowest scale, we simply compute the active vertices by point location for $P$
in a cubical grid, and enlist $n$ inclusion events
(this is the only step where the input points are considered in the algorithm).
We use an auxiliary container $S$ and maintain the invariant
that whenever a new scale is considered, 
$S$ consists of all simplices of the previous scale,
sorted by dimension. 
In $S$, for each $0$-simplex, we store an id and
a coordinate representation of the active face to which it corresponds.
Every $\ell$-simplex with $\ell>0$ is stored just as a list of integers, 
denoting its boundary $0$-simplices.
We initialize $S$ with the $n$ $0$-simplices at the lowest scale.

Let $\alpha<\alpha'$ be any two consecutive 
scales with $\square,\square'$ the respective cubical complexes
and $\ux,\ux'$ the approximation complexes, with $\tg:\ux\rightarrow \ux'$ 
being the simplicial map connecting them. 
Suppose we have already constructed all events at scale $\alpha$. 
We enlist the scale event for $\alpha'$.
Then, we enlist the contraction events. 
For that, we iterate through the $0$-simplices
of $\ux$ and compute their value under $g$, using point location in a cubical grid.
We store the results in a list $S'$ (which contains the simplices of $\ux'$).
If for a $0$-simplex $j$, $g(j)$ is found to be equal to $g(i)$ for a 
previously considered $0$-simplex, we choose the minimal such $i$ and 
enlist a contraction event for $i$ and $j$.

We turn to the inclusion events and start with the case of $0$-simplices.
Every $0$-simplex is an active face at scale $\alpha'$ and must contain an 
active vertex, which is also a $0$-simplex of $\ux'$. 
We iterate through the elements in $S'$. 
For each active vertex $v$ encountered, we go over all faces of the cubical complex 
$\square'$ that contain $v$ as vertex and check whether they are active. 
For every active face encountered that is not in $S'$ yet, we add it to $S'$
and enlist an inclusion event of a new $0$-simplex. 
At termination, all $0$-simplices of $\ux'$ have been detected.

Next, we iterate over the simplices of $S$ of dimension $\geq 1$ and 
compute their image under $\tg$, and store the result in $S'$. 
To find the simplices of dimension $\geq 1$ included at $\ux'$, we exploit our 
previous insight that they contain at least one $0$-simplex that is included 
at the same scale (see the proof of Theorem~\ref{theorem:towersize}).
Hence, we iterate over the $0$-simplices included in $\ux'$ and 
proceed inductively in dimension.
Let $v$ be the current $0$-simplex under consideration;
assume that we have found all $(p-1)$-simplices in $\ux'$ that contain $v$. 
Each such $(p-1)$-simplex $\sigma$ is a flag in $\square'$. 
We iterate over all faces $e$ that extend $\sigma$ to a flag of length $p+1$.
If $e$ is active, we found a $p$-simplex in $\ux'$. 
If this simplex is not in $S'$ yet, we add it and enlist an inclusion event for it. 
We also enqueue the simplex in our inductive procedure, to look for
$(p+1)$-simplices in the next iteration. 
At the end of the procedure, we have detected all simplices in $\ux'$ 
without preimage, and $S'$ contains all simplices of $\ux'$. 
We set $S\gets S'$ and proceed to the next scale. 
This ends the description of the algorithm.

\begin{theorem}
\label{theorem:algo1_bary}
To compute the $k$-skeleton, 
the algorithm takes time $\big(n2^{O(d)}\log\Delta + 2^{O(d)}M\big)$ time 
in expectation and $M$ space, where $M$ is the size of the tower. 
In particular, the expected time is bounded by 
$\big(n2^{O(d)}\log\Delta + n2^{O(d\log k)}\big)$
and the space is bounded by $n2^{O(d\log k)}$.
\end{theorem}

\begin{proof}
In the analysis, we ignore the costs of point locations in grids,
checking whether a face for being active, and searches in data structures $S$,
since all these steps have negligible costs when
appropriate data structures are chosen.

Computing the image of a $0$-simplex of $\ux$ costs $O(2^d)$ time.
Moreover, there are at most $n2^{O(d)}$ vertices altogether in the tower,
so this bound in particular holds on each scale.
Hence, the contraction events on a fixed scales can be computed in $n2^{O(d)}$.
Finding new $0$-simplices requires iterating over the cofaces of a vertex 
in a cubical complex. 
There are $3^d$ such cofaces. 
This has to be done for a subset of the $0$-simplices
in $\ux'$, so the running time is also $n2^{O(d)}$. 
Since there are $O(\log\Delta+\log d)$ scales considered, 
these steps require $n2^{O(d)}\log\Delta$ over all scales.

Computing the image of $\tg$ for a fixed scale costs at most $O(2^d|\ux|)$.
$M$ is the size of the tower, that is, the simplices without preimage, 
and $I$ the set of scales considered, so
the expected bound for $\sum_{\alpha\in I} |\ux_\alpha|=O(\log d M)$, 
because every simplex
has an expected lifetime of at most $3\log d$ by Lemma~\ref{lemma:g_survival}. 
Hence, the cost of these steps is bounded by $2^{O(d)}M$.

In the last step of the algorithm, we consider a subset of simplices of $\ux'$. 
For each one, we iterate over
a collection of faces in the cubical complex of size at most $2^O(d)$. 
Hence, this step is also bounded by $O(2^{O(d)}|\ux|)$ per scale, 
and hence bounded $2^{O(d)}M$ as well.

For the space complexity, the auxiliary data structure $S$ gets as large as $\ux$, 
which is clearly bounded by $M$.
For the output complexity, the number of contraction events is smaller than the 
number of inclusion events, because every contraction removes a $0$-simplex that 
has been included before.
The number of inclusion events is the size of the tower. 
The number of scale events as described is $O(\log\Delta+\log d)$. 
However, it is simple to get rid of this factor by only including scale events 
in the case that at least one inclusion/contraction
takes place at this scale. 
The space complexity bound follows.
\end{proof}

\section{Conclusion}
\label{section:conclusion}
We gave an approximation scheme for the Rips filtration, with improved
approximation ratio, size and computational complexity than previous approaches
for the case of high-dimensional point clouds.
Moreover, we introduced the technique of using acyclic carriers to prove
interleaving results. 
We point out that, while the proof of the interleaving in
Section~\ref{subsection:intlv} is still technically challenging,
it greatly simplifies by the usage of acyclic carriers; 
defining the interleaving chain maps
explicitly significantly blows up the analysis. 
There is also no benefit in knowing the interleaving maps because they are only 
required for the analysis, not for the computation.

Our tower is connected by simplicial maps; there are (implemented) algorithms
to compute the barcode of such towers~\cite{desu-gic,kstwr}.
It is also quite easy to adapt our tower construction to a streaming 
setting~\cite{kstwr}, where the output list of events is passed to an 
output stream instead of being stored in memory.

An interesting question is whether persistence can be computed efficiently 
for more general chain maps, which would allow more freedom in building 
approximation schemes.

\subparagraph*{Acknowledgments} Michael Kerber is supported by the 
Austrian Science Fund (FWF) grant number P 29984-N35. 
Sharath Raghvendra acknowledges support of NSF CRII grant CCF-1464276.

\bibliographystyle{plain}

\end{document}